\documentclass[journal,twoside,web]{ieeeconf}
\usepackage{amsmath,amssymb,amsfonts}
\usepackage[dvipsnames]{xcolor}
\usepackage{mathtools}
\usepackage{floatrow}
\usepackage{tikz}
\usepackage{url}
\usepackage{hyperref}
\usepackage{verbatim}   
\usepackage{ntheorem}
\usepackage{algorithm,algorithmic}

\usepackage{siunitx}
\usepackage{tabularx}
\usepackage{pgfplots}
\usepackage{pgfplotstable}
\usepackage{multirow}
\usepgfplotslibrary{fillbetween}
\usepackage{booktabs}
\pgfplotsset{compat = newest}
\usetikzlibrary{arrows,positioning,shapes,intersections,external,patterns,calc,fit,decorations,decorations.markings} 
\usepackage{cite}

\newtheorem{theorem}{Theorem}

\newtheorem{remark}{Remark}

\newtheorem{lemma}{Lemma}

\newcommand{\R}{\mathbb{R}}

\usepackage{stmaryrd}
\newcommand{\llangle}{\langle\!\langle}
\newcommand{\rrangle}{\rangle\!\rangle}

\newcommand{\N}{\mathbb{N}}

\newcommand{\D}{\mathcal{D}}
\newcommand{\X}{\mathcal{X}}

\newcommand{\bm}[1]{{\boldsymbol{#1}}}
\newcommand{\Verts}[1]{{\left\Vert #1 \right\Vert}}
\DeclareMathOperator{\diag}{diag}
\DeclareMathOperator{\var}{var}
\DeclareMathOperator{\mean}{\mu}
\DeclareMathOperator{\Var}{\Sigma}
\DeclareMathOperator{\Mean}{\bm\mu}

\DeclareMathOperator{\Prob}{P}

\begin{document}
\title{Learning-based decentralized control with collision avoidance for multi-agent systems*}
\author{Omayra Yago Nieto$^{1}$, Alexandre Anahory Simoes$^{2}$, Juan I. Giribet$^{3}$ and Leonardo J. Colombo$^{4}$ 
\thanks{$^{1}$O.Yago Nieto is a PhD Student at Universidad Politécnica de Madrid, Spain. email:  omayra.yago.nieto@alumnos.upm.es}
\thanks{$^{2}$A.Anahory Simoes is with the School of Science and Technology, IE University, Spain. email:  alexandre.anahory@ie.edu}
\thanks{$^{3}$ J. Giribet is with LINAR Laboratory, Universidad de San Andr\'es, Buenos Aires, and CONICET, Argentina. email: jgiribet@conicet.gov.ar}
\thanks{$^{4}$ L. Colombo is with Centre for Automation and Robotics (CSIC-UPM), Centre for Automation and Robotics (CSIC-UPM), Ctra. M300 Campo Real, Km 0,200, Arganda
del Rey - 28500 Madrid, Spain. email:  leonardo.colombo@csic.es}
\thanks{* J. Giribet was partially supported by  PICT-2019-2371 and PICT-2019-0373 projects from Agencia Nacional de Investigaciones Cient\'ificas y Tecnol\'ogicas, and UBACyT-0421BA project from the Universidad de Buenos Aires (UBA), Argentina. The authors acknowledge financial support from the Spanish Ministry of Science and Innovation, under grants  PID2022-137909NB-C21 funded by MCIN/AEI\-/10.13039\-/501100011033}}

\maketitle

\begin{abstract}
 In this paper, we present a learning-based tracking controller based on Gaussian processes (GP) for collision avoidance of multi-agent systems where the agents evolve in the special Euclidean group in the space SE(3). In particular, we use GPs to estimate certain uncertainties that appear in the dynamics of the agents. The control algorithm is designed to learn and mitigate these uncertainties by using GPs as a learning-based model for the predictions. In particular, the presented approach guarantees that the tracking error remains bounded with high probability. We present some simulation results to show how the control algorithm is implemented.
\end{abstract}



\section{Introduction}
Extending the concept of a single autonomous mobile robot performing a task to a group of robots has been an area of active research in the last few decades. One of the key elements in the operation of groups of mobile robots is the control method used to coordinate the behavior of each robot~\cite{jorgebook}.  Some of the most widely used concepts are based on virtual potential fields \cite{Leonard}, \cite{RK}. Decentralized navigation approaches are more appealing to centralized ones, due to their reduced computational complexity and increased robustness concerning agent failures. In \cite{dimarogonas2003decentralized} and \cite{dimos_collisionDI05}, the methodology of the navigation function established in \cite{RK}
for centralized navigation of multiple robots is
extended to decentralized navigation and planar (2D) double integrator agents. In this work, we first extend this control scheme for dynamical agents moving in 3D space whose configuration is the Lie group $SE(3)$. Next, we employ such a result for learning-based control of multi-agent systems with partially unknown dynamics avoiding collision between them.

The design of safe controllers for multi-agent systems is a substantial aspect for an increasing range of application domains. However, parts of the robot's dynamics and external disturbances are often unknown or time-consuming to model. 
To overcome the issue of unknown dynamics, learning-based control laws have been proposed but they are limited to iterative learning~\cite{liu2012iterative}, leader-follower formations~\cite{yuan2017formation,yan2021,8870252} or lack of guarantees~\cite{Schilling2019}.  We present a safe decentralized controller for agents with dynamics evolving in the Lie group $SE(3)$ using Gaussian processes to learn the unknown dynamics, while agents avoid collisions between them. More concretely, we derive a set of data-driven control laws (one for each agent)
that drives a team of agents with partially unknown dynamics from any initial configuration to
a desired goal configuration, avoiding, at the same time,
collision between agents. The environment is assumed to be perfectly known and stationary, while each agent has knowledge of the distance to its neighbors and the dataset based on its own dynamics containing noisy data collected online relative to its position, attitude, and measured forces and torque. The learning-based controller theoretically guarantees that the trajectory of the system converges to the desired configuration with a bounded error, with probability at least $\epsilon$, where the bound is explicitly given. The online learning approach based on GPs allows to improve the model and, thus, the stability performance during runtime.

Decentralized machine learning control algorithms for multi-agent systems has been recently studied in~\cite{Tol2020, Gama2021} using graph neural networks, and also in \cite{yang2021distributed}. In particular, recent learning-based methods for formation control and flocking control of second-order planar agents can be found in \cite{beckers2021online}, \cite{beckers2022learning}, but none of them includes dynamical agents on $SE(3)$, with stability guarantees in the performance. To the best of our knowledge, there are no results for the design of a learning-based decentralized control law for dynamic agents on $SE(3)$ under partially unknown dynamics guaranteeing a bounded tracking error within a given probability.

The remainder of the paper is structured as follows. In Section \ref{sec2}, we present preliminary material on Riemannian geometry and Lie groups that we will use in the paper. In Section \ref{sec3}, we design the nominal control law for asymptotic stability of the agents on $SE(3)$ avoiding collision by using decentralized navigation functions. In Section \ref{sec4}, the online learning-based controller is given and in Section \ref{sec5}, we present some simulation results. 

\section{Dynamics of agents evolving on $SE(3)$}\label{sec2}

This section introduces conventional mathematical notions to describe simple mechanical systems on the Lie group $SE(3)$ which can be found in \cite{bullo2019geometric}, \cite{mars}, and \cite{BL} for instance. 

\subsection{Background on Differential Geometry}\label{gm}
Let $Q$ be a differentiable manifold with $\hbox{dim}(Q)=n$. Throughout the text, $q^{i}$ will denote a particular choice of local coordinates on this manifold, and $TQ$ denotes its tangent bundle, with $T_{q}Q$ denoting the tangent space at a specific point $q\in Q$. Usually $v_{q}$ denotes a vector at $T_{q}Q$ and, in addition, the coordinate chart $q^{i}$ induces a natural coordinate chart on $TQ$, denoted by $(q^{i},\dot{q}^{i})$ with $\hbox{dim}(TQ)=2n$. Let $T^{*}Q$ be its cotangent bundle, locally described by the positions and the momentum for the system, i.e., $(q,p)\in T^{*}Q$ with $\hbox{dim}(T^{*}Q)=2n$. The cotangent bundle at a point $h\in Q$ is denoted as $T_{h}^{*}Q$.

A Riemannian manifold is denoted by the 2-tuple $(Q, \left< \cdot, \cdot\right>)$, where $Q$ is a smooth connected manifold and $\left< \cdot, \cdot\right>$ is the metric on $Q$, a positive-definite symmetric covariant 2-tensor field. That is, to each point $q\in Q$, we assign a positive-definite inner product $\left<\cdot, \cdot\right>_q:T_qQ\times T_qQ\to\mathbb{R}$, where $\left<\cdot, \cdot\right>_q$ varies smoothly with respect to $q$. The length of a tangent vector is determined by its norm, defined by $\Vert v_q\Vert=\left<v_q,v_q\right>^{1/2}$ with $v_q\in T_qQ$. 

For any $p \in Q$, the Riemannian metric induces an invertible map $\cdot^\flat: T_p Q \to T_p^\ast Q$, called the \textit{flat map}, defined by $X^\flat(Y) = \left<X, Y\right>$ for all $X, Y \in T_p Q$. The inverse map $\cdot^\sharp: T_p^\ast Q \to T_p Q$, called the \textit{sharp map}, is similarly defined implicitly by the relation $\left<\alpha^\sharp, Y\right> = \alpha(Y)$ for all $\alpha \in T_p^\ast Q$.


 Let $G$ be a finite dimensional Lie group with identity element $e\in G$. A \textit{left-action} of $G$ on a manifold $Q$ is a smooth mapping $\Phi:G\times Q\to Q$ such that $\Phi(e,q)=q$ for all $q\in Q$, $\Phi(g,\Phi(h,q))=\Phi(gh,q)$ for all  $g,h\in G, q\in Q$, and for every $g\in G$, $\Phi_g:Q\to Q$ defined by $\Phi_{g}(q):=\Phi(g,q)$, is a diffeomorphism. Let $\mathfrak{g}$ be the Lie algebra associated to $G$, that is, $\mathfrak{g}=T_{e}G$.  The infinitesimal generator corresponding to $\xi \in \mathfrak{g}$ is the vector field on $TQ$, denoted by $\xi_Q$, which is defined as $\xi_Q(q) = \frac{\mathrm{d}}{\mathrm{d}t}|_{t=0} \Phi (\exp(t\xi), q)$, where $\exp$ denotes the exponential map. 
 
 The Lie bracket on $\mathfrak{g}$ is denoted by $[\cdot,\cdot]$. The \textit{adjoint map}, $ad_\xi : \mathfrak{g} \to \mathfrak{g}$ for $\xi \in \mathfrak{g}$ is defined as $ad_\xi \eta \coloneq [\xi, \eta]$ for $\eta \in \mathfrak{g}$. We denote by $L: G \times G \to G$ the \textit{left group action} in the first argument defined as $L(g,h) = L_{g} (h)= gh $ for all $g$, $h \in G$.
 
 Let $\mathbb{I} :\mathfrak{g} \to \mathfrak{g}^*$ be an isomorphism from the Lie algebra to its dual. The \textit{inverse} is denoted by $\mathbb{I}^\sharp: \mathfrak{g}^* \to \mathfrak{g}$. The isomorphism $\mathbb{I}$ induces the inner product $\llangle\cdot,\cdot\rrangle_{\mathbb{I}}:\mathfrak{g}\times\mathfrak{g}\to\mathbb{R}$ given by $\llangle\xi_{1},\xi_{2}\rrangle_{\mathbb{I}} = \langle \mathbb{I}(\xi_{1}), \xi_{2} \rangle_{\mathfrak{g}}$, for $\xi_1,\xi_2\in\mathfrak{g}$ and where $\langle\cdot,\cdot\rangle_{\mathfrak{g}}:\mathfrak{g}^{*}\times\mathfrak{g}\to\mathbb{R}$ denotes the natural pairing between elements of $\mathfrak{g}^{*}$ and $\mathfrak{g}$.  In addition, $\mathbb{I}$ induces a
left invariant metric on $G$ (see \cite{bullo2019geometric}), which we denote by $\mathbb{G}_{\mathbb{I}}$ and is defined by  $\mathbb{G}_{\mathbb{I}}(g).(X_g,Y_g) = \langle \mathbb{I}(T_gL_{g^{-1}} (X_g)),T_gL_{g^{-1}} (Y_g)\rangle$ for all $g \in G$ and $X_g$, $Y_g \in T_g G$. 

\subsection{Modelling of agents on SE(3)}
Let $SE(3)$ be the special Euclidean group on the space. Any $g \in SE(3)$ is represented as 
$g=\begin{pmatrix}R & q \\
 0 & 1
\end{pmatrix}$ where $R\in SO(3)$, the special orthogonal group in the space, and $q$ is the position in $\mathbb{R}^3$.
The group operation is the matrix multiplication. The Lie algebra of $SE(3)$ is denoted by $\mathfrak{se}(3)$ and any $\xi \in \mathfrak{se}(3)$ is represented as $\xi = \begin{pmatrix}
           \hat{\Omega} & v \\
           0 & 0
         \end{pmatrix}$, where

         \[ \Omega= \begin{pmatrix}
                                   \omega_1 \\ \omega_2 \\ \omega_3
                                 \end{pmatrix}, \,\hat{\Omega}= \begin{pmatrix}
                                   0 & -\omega_3 & \omega_2 \\
                                   \omega_3 & 0 & -\omega_1 \\
                                   -\omega_2 & \omega_1 & 0
                                 \end{pmatrix}, \, v = \begin{pmatrix}
                                                             v_1 \\ v_2 \\ v_3
                                                           \end{pmatrix}.
\]
With this notation, an element of $\mathfrak{se}(3)$ will be sometimes denoted as $\xi=(\hat{\Omega},v)$ for $\hat{\Omega}\in\mathfrak{so}(3)$ and $v\in\mathbb{R}^3$, where $\mathfrak{so}(3)$ denotes the set of $3\times 3$ skew-symmetric matrices. 

We define two projection maps $\pi_1: SE(3) \to SO(3)$ and $\pi_2 : SE(3) \to \mathbb{R}^3$ by $\pi_1(g)= \begin{pmatrix}
               I_{3 \times 3} & 0_{3\times 1}
             \end{pmatrix}g \begin{pmatrix}
               I_{3 \times 3} \\
               0_{1 \times 3}
             \end{pmatrix}$ and 
             $\pi_2(g)= \begin{pmatrix}
               I_{3 \times 3} & 0_{3\times 1}
             \end{pmatrix}g \begin{pmatrix}
                              0_{3\times 1} \\
                              1
                            \end{pmatrix}$.

It should be noted that these maps do not depend on the choice of coordinates in $SO(3)$. 

The \textit{adjoint map} $ad_{(\hat{\eta},v)} : \mathfrak{se}(3) \to \mathfrak{se}(3)$ is defined as $$ad_{(\hat{\eta},v)}(\hat{\xi},\bar{v})=(\hat{\eta}\times\hat{\xi}, \hat{\eta}\times\bar{v}-\hat{\xi}\times v),$$ for $ (\hat{\eta},v), (\hat{\xi},\bar{v}) \in \mathfrak{se}(3)$.
The \textit{dual of the adjoint map} $ad^*:\mathfrak{se}(3)\times\mathfrak{se}(3)^* \to\mathfrak{se}(3)^*$ is defined as 
$$ad^{*}_{(\xi,\alpha)}(\mu,\beta)=(\mu\times\xi-\alpha\times\beta,-\xi\times\beta),$$ for $(\xi,\alpha)\in\mathfrak{se}(3)$ and $(\mu,\beta)\in\mathfrak{se}(3)^{*}$.

Next, consider a multi-agent system with $s$ agents where each agent is modeled as a fully-actuated mechanical control system on $SE(3)$. This means we have $6$ actuators on each agent denoted by $u_i$, with $i=1,\ldots,s$. The dynamics of agent $i$ are given by the Euler-Poincar\'e equations
\begin{align}\label{dynliegrp}
\xi_i &= T_{g_i} L_{g^{-1}_i} \dot{g}_i,\\ \nonumber
\dot{\xi}_i&= u_i+\mathbb{I}^\sharp ad^*_{\xi_i} (\mathbb{I} \xi_i),
\end{align} for $i=1,\ldots,s$ and $\mathbb{I} :\mathfrak{se}(3) \to \mathfrak{se}(3)^*$ is an isomorphism from the Lie algebra $\mathfrak{se}(3)$ to its dual $\mathfrak{se}(3)^{*}$, with inverse map denoted by $\mathbb{I}^\sharp$ defined in subsection \ref{gm}.


\section{Nominal control law for collision avoidance of agents on $SE(3)$}\label{sec3}
In a nutshell, the control action developed in this paper is the negative gradient of a potential function. The potential function has two roles: (i) it takes very high values when there is a chance of a possible collision, and (ii) it has a unique minimum at the target configuration. Next, we describe the construction of this potential function and the control law to avoid collision among dynaical agents on $SE(3)$.

\subsection{Collision avoidance task}
Let the target configuration for agent $i$ be denoted by $g_{di}$. The proposed potential function for agent $i$ is
\begin{equation}\label{potfun}
  \psi_i = \frac{\gamma_{di} + f_{i} }{((\gamma_{di} + f_{i})^k + G_i)^{1/k}}
\end{equation}
where $\gamma_{di}(R_{i},q_{i}) = \hbox{trace}(R_i R^T_{di}-I_{3 \times 3}) + \| q_{i} - q_{di}\|^{2}$, with $I_{3 \times 3}$ being the identity of $SE(3)$. The first term is an approximation of the Riemannian distance between $R_{i}$ and $R_{di}$, and the second term is the Euclidean distance in $\R^{3}$.

The function $G_i$ expresses all possible collisions of agent $i$ with the others. The exponent $k$ is a scalar positive parameter. The function $\psi_i$ is the potential function with a minimum at $g_{di}$. The term $f_{i}$ is a correction term that allows to temporarily move the minimum away from $g_{di}$ whenever collisions with agent $i$ tend to occur. The inclusion of this term allows agent $i$ to temporarily depart from its minimum in order to avoid collisions.

In order to define $G_i$, we first state the possible collisions which can occur.  The term \textit{relation} is used to describe the possible collision schemes that can occur with respect to agent $i$. A \textit{binary relation} is a relation between agent $i$ and another. Each collision comprises one or many binary relations. The number of binary relations in a relation is called its \textit{level}. There are several possible relations at a certain level, these are referred to as first, second, etc. 


The \textit{relation proximity function} (RPF) between two agents is defined by 
\[ \beta_{ij} = ||\pi_2(g_i) - \pi_2(g_j) ||^2 - {(r_i + r_j)}^2
\]
which vanishes as agents $i$ and $j$ tend to collide. Here $g_i,g_j\in SE(3)$ and $r_i,r_j\in\mathbb{R}_{\geq 0}$. A RPF provides a measure of the distance between agent $i$ and the other agents involved in the relation. Each relation has its own RPF.

The RPF of the $k^{th}$ relation of level $l$ is $\displaystyle{{(b_{R_k})}_l = \sum_{j \in {R_k}_l} \beta_{ij}}$ where $R_k$ denotes the $k^{th}$ relation of level $l$. $h_{R_j}$ is the \textit{relation verification function} (RVF), which is defined as
\[ {(h_{R_k})}_l = {(b_{R_k})}_l + \frac{\lambda {(b_{R_k})}_l}{ {(b_{R_k})}_l + {(B_{R_k^C})}_l^{1/\sigma}}
\]
where $\lambda$, $\sigma$ are positive scalars and
\[ {(B_{R_k^C})}_l = \prod_{m \in R^C_k} {(b_m)}_l
\]
and ${({R_k^C})}_l$ is the complementary set of relations of level-$l$, that is, the set of all other relations in level $l$ other than the $k^{th}$ relation, or in other words, all the other relations with respect to agent $i$ that have the same number of binary relations with the relation $R_k$. Note that for the highest level $l=s-1$, only one relation is possible and so $(R_k^{C})_{s-1}=\varnothing$.

The key property of RVFs is that the RVF of one and only one relation can tend to zero at each time instant, namely, the RVF of the relation that holds at the highest level. A relation holds when
the proximity functions of all its binary relations tend to zero. Hence it serves as an analytic switch that is activated (tends to zero) only when the relation it represents is realized.



The function $G_i$ is defined as $\displaystyle{G_i = \prod_{l=1}^{n_L^i}\prod_{j=1}^{n_{R_l}^i} {(h_{R_j})}_l}$ where $n_L^i$ is the number of levels, and $n_{R_l}^i$ is the number of relations in level-$l$ with respect to agent $i$. Hence $G_i$ is the product of the RVF, of all relations with respect to $i$. 


\begin{remark}The definition of the function $G_i$ in the multiple moving agents situation is slightly different than the one introduced by the authors in \cite{RK}. The collision avoidance scheme in that approach involved a single moving point agent in an environment with static obstacles. A collision among rigid bodies and with more than one obstacle was therefore impossible, and the obstacle function was simply the product of the distances of the agent from each obstacle.\end{remark}

As we mentioned before, and following the same reasoning as in \cite{dimos_collisionDI05}, the function $f_i$ is defined by
\begin{equation*}
    f_{i}(G_{i}) = \begin{cases}
            & a_{0} + \sum_{j=1}^{3} a_{j}G_{i}^{j}, G_{i} \leqslant X \\
            & 0, G_{i}>X
        \end{cases}
\end{equation*}
where $X, a_{j}$ are parameters with $X>0$ and $a_{1} = 0$, $a_{2} = \frac{-3 a_{0}}{X^{2}}$, $a_{3} = \frac{2 a_{0}}{X^{3}}.$

Note that $a_{0}=f_{i}(0)$ is also a parameter. The parameter $X$ should be regarded as a sensitivity parameter that flags the possibility of collision occurrence. We will require that all desired positions satisfy: $G_{i}(g_{1}, \dots, g_{s})>X$, where $s$ is the number of agents.

\subsection{Nominal Control Law}\label{nominalsec}

Next, we present the decentralized control law for collision avoidance among dynamical agents on $SE(3)$.

\begin{theorem}\label{theorem}
  The following decentralized feedback control law guarantees asymptotic convergence to $g_{di}$ for every agent $i$ and avoids possible collisions between them
  \begin{align*} u_i &= - \mathbb{I}^{\sharp} \left(K_i T_{g_i} L_{g_i^{-1} } \left(\frac{\partial \psi_i}{\partial g_i}\right)\right) + \xi_i \theta_i\left(\xi_i, \frac{\partial \psi_i}{\partial t}\right) \\&+ F_d \xi_i - \mathbb{I}_i^\sharp ad^*_{\xi_i}( \mathbb{I}_i \xi_i),
  \end{align*}
  where, $\xi_i = g^{-1}_i(t) g_i(t)$, $K_i \in \mathbb{R}^+$,  $F_d$ is a dissipative $(1,1)$-tensor, i.e., satisfying ${\llangle F_d \xi_i , \xi_i \rrangle}_{\mathbb{I}_i}\leqslant 0$, $\theta_i(\xi_i, \frac{\partial \psi_i}{\partial t}) \in \mathbb{R}$ is given by $\displaystyle{\theta_i\left(\xi_i, \frac{\partial \psi_i}{\partial t}\right)=-\frac{c}{\tanh({\llangle \xi_i, \xi_i\rrangle}_{\mathbb{I}_i})} \Big| \frac{\partial \psi_i}{\partial t} \Big|}$, where $c > \hbox{max}_i\{K_i\}$, and, $\displaystyle{\frac{\partial \psi_i}{\partial t} \in \mathbb{R}}$ is defined as $\displaystyle{\frac{\partial \psi_i}{\partial t} = \sum_{j \neq i}\Big\langle \frac{\partial \psi_i}{\partial g_j}, \dot{g}_j \Big\rangle}$.\end{theorem}
\begin{proof}
  Let $\displaystyle{V= \sum_{i=1}^{s} K_i \psi_i + \frac{1}{2} {\llangle \xi_i, \xi_i\rrangle}_{\mathbb{I}_i}}$. So the time derivative of the Lyapunov function along $u_i$ is given by 
  \begin{align*} 
    \dot{V} &= \sum_{i=1}^{s} K_i \left( \frac{\partial \psi_i}{\partial t} +  \Big\langle  \frac{\partial \psi_i}{\partial g_i},  {\dot{g}}_i \Big\rangle \right) + {\llangle \dot{\xi}_i, \xi_i \rrangle}_{\mathbb{I}_i}\\
    &=\sum_{i=1}^{s} K_i \left( \frac{\partial \psi_i}{\partial t} +  \Big\langle T^{*}_{g_i} L_{g_i^{-1}}  \frac{\partial \psi_i}{\partial g_i},   T_{g_i} L_{g_i^{-1}} {\dot{g}}_i \Big\rangle \right) + {\llangle \dot{\xi}_i, \xi_i \rrangle}_{\mathbb{I}_i}\\
    &= \sum_{i=1}^{s} K_i \left( \frac{\partial \psi_i}{\partial t} +  \Big\langle T_{g_i} L_{g_i^{-1}}  \frac{\partial \psi_i}{\partial g_i},   \xi_i \Big\rangle \right)\\
    &+{\Big\langle\Big\langle -K_i \mathbb{I}^{\sharp} \left(T_{g_i} L_{g_i^{-1}} \frac{\partial \psi_i}{\partial g_i} \right)+ \xi_i \theta_i\left(\xi_i, \frac{\partial \psi_i}{\partial t}\right) + F_d \xi_i, \xi_i \Big\rangle\Big\rangle}_{\mathbb{I}_i}\\
    &= \sum_{i=1}^{s} K_i  \frac{\partial \psi_i}{\partial t} + {\llangle \theta_i \xi_i , \xi_i \rrangle}_{\mathbb{I}_i} + {\llangle F_d \xi_i , \xi_i \rrangle}_{\mathbb{I}_i}.
  \end{align*}
  Observe that, if $\displaystyle{\frac{\partial \psi_i}{\partial t} > 0}$, as $c > \underset{i}{max} \{K_i\}$, then  $$ c> K_i \frac{\tanh {\llangle \xi_i, \xi_i\rrangle}_{\mathbb{I}_i}}{{\llangle \xi_i, \xi_i\rrangle}_{\mathbb{I}_i}} \qquad \forall i.$$ Therefore,  
$$ K_i \frac{\partial \psi_i}{\partial t} <  \frac{c {\llangle \xi_i, \xi_i\rrangle}_{\mathbb{I}_i}}{\tanh {\llangle \xi_i, \xi_i\rrangle}_{\mathbb{I}_i}} \big| \frac{\partial \psi_i}{\partial t} \big|,$$ so 
  $$K_i \frac{\partial \psi_i}{\partial t} + {\llangle \theta_i \xi_i, \xi \rrangle}_{\mathbb{I}_i} < 0.$$

  If $\displaystyle{\frac{\partial \psi_i}{\partial t}<0}$, $c >0$, then $$c > -K_i \frac{\tanh {\llangle \xi_i, \xi_i\rrangle}_{\mathbb{I}_i}}{{\llangle \xi_i, \xi_i\rrangle}_{\mathbb{I}_i}} \quad\forall i,$$ so
 $$-K_i \big| \frac{\partial \psi_i}{\partial t}\big|  -\frac{c {\llangle \xi_i, \xi_i\rrangle}_{\mathbb{I}_i}}{\tanh {\llangle \xi_i, \xi_i\rrangle}_{\mathbb{I}_i}}  \big| \frac{\partial \psi_i}{\partial t}\big| <0,$$ and therefore
$$ K_i \frac{\partial \psi_i}{\partial t} + {\llangle \theta_i \xi_i, \xi \rrangle}_{\mathbb{I}_i} < 0.$$
In addition, the function $\displaystyle{K_i \frac{\partial \psi_i}{\partial t} + {\llangle \theta_i \xi_i, \xi \rrangle}_{\mathbb{I}_i} = 0}$ when $\displaystyle{\frac{\partial \psi_i}{\partial t} =0}$. 

Hence, by LaSalle's Invariance Principle, the state
of the system converges to the largest invariant set
contained in 
\begin{align*}S &= \Big\{ (g, \xi)\in SE(3)\times\mathfrak{se}(3) \Big| \frac{\partial \psi_i}{\partial t} =0 \text{ and } \xi_{i}=0 \Big\}\\&= \{ (g, \xi)\in SE(3)\times\mathfrak{se}(3) |  \xi_{i}=0\},\end{align*}
where the last equality is implied by the definition of $\frac{\partial \psi_i}{\partial t}$. The largest invariant set contained in $S$ is characterized by $\displaystyle{\frac{\partial \psi_{i}}{\partial g_{i}} = 0}$ for each $i$. The desired positions $g_{di}$ satisfy this condition and are isolated critical points of $\psi_{i}$ \cite{koditschek1990robot}. The set of initial conditions that lead to saddle points are sets of measure zero \cite{milnor1963morse}. Hence the largest invariant set contained in the set $\displaystyle{\frac{\partial \psi_{i}}{\partial g_{i}} = 0}$ for each $i$ is $g_{di}$.
\end{proof}


\section{Learning-based control with Gaussian processes}\label{sec4}

Next, consider a multi-agent system with $s$ agents. Each agent is modeled as a fully-actuated mechanical system on $SE(3)$ and subject to unknown dynamics denoted by $\mathbf{f}_{uk}^{i}:SE(3)\times\mathfrak{se}(3)\to\mathfrak{se}(3)^{*}$. The dynamics of agent $i$ are given by the Euler-Poincaré equations
\begin{align}\label{for:system1}
\xi_i &= T_{g_i} L_{g^{-1}_i} \dot{g}_i,\\ \nonumber
\dot{\xi}_i-\mathbb{I}_i^\sharp ad^*_{\xi_i} (\mathbb{I}_i \xi_i)&= u_i+\mathbf{f}_{uk}^{i}(\mathbf{p}_i),
\end{align} for $i=1,\ldots,s$, where $\mathbf{p}_i=(g_i,\xi_i)=((R_i,q_i),(\Omega_i,v_i))\in SE(3)\times\mathfrak{se}(3)\simeq SE(3)\times\mathbb{R}^{6}$. As introduced in the vehicle's dynamics~\eqref{for:system1}, we assume that parts of the dynamics are unknown, i.e., $\mathbf{f}_{uk}^{i}$. The proposed control strategy is based on the design of a decentralized controller by using a model that is updated by the predictions of Gaussian Processes (GP). 

The data for the GP is collected in arbitrary time intervals of the vehicle's dynamics during the control process. Then, the predictions of the GP are updated based on the collected dataset and the vehicle model is improved. An overview of the proposed decentralized control strategy for each agent is depicted in Figure \ref{fig:control_graph}.
\begin{figure}
    \centering
    \includegraphics[width=0.9\textwidth]{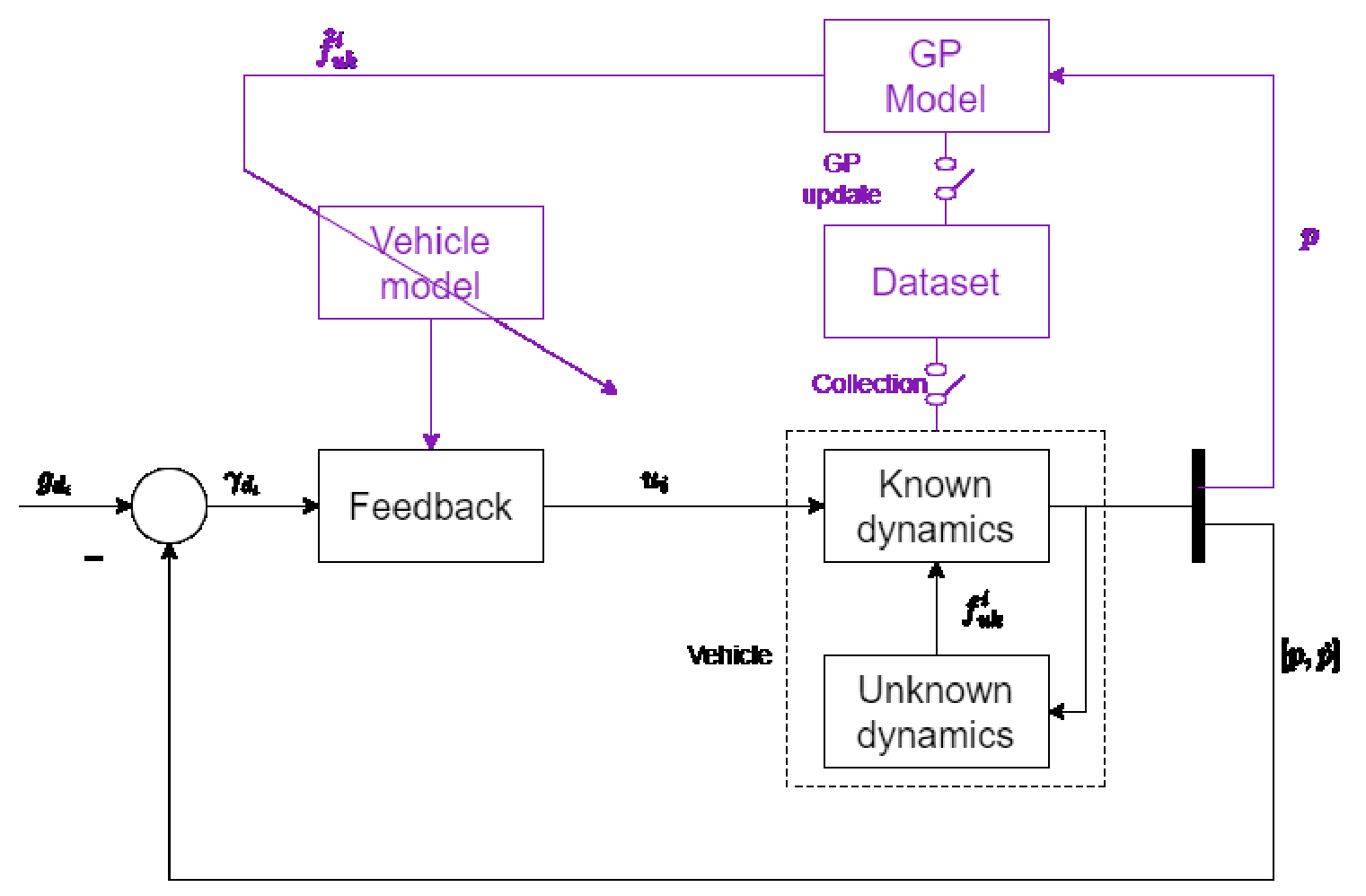}
    \caption{Block diagram of the proposed decentralized control law for each agent.}
    \label{fig:control_graph}
\end{figure} Next, we present the learning and control framework in detail.

\subsection{Learning with Gaussian processes}
For the compensation of the unknown dynamics of \eqref{for:system1}, we use  Gaussian Processes (GPs) to estimate the values of $\mathbf{f}_{uk}^{i}$ for a given state $\mathbf{p}_i$. For this purpose, $N(n):\mathbb{N}\to\mathbb{N}$ training points of the system \eqref{for:system1} are collected to create a dataset
\begin{align}
\D_{i,n(t)}=\{\mathbf{p}_i^{\{j\}},\mathbf{y}_i^{\{j\}}\}_{j=1}^{N(n)},\,\ i=1,\ldots s.\label{for:dataset}
\end{align}
 The output data $\mathbf{y}\in\mathfrak{se}(3)^{*}\simeq\R^6$ is given by $\mathbf{y}=[\dot{\xi}_i-\mathbb{I}_i^\sharp ad^*_{\xi_i} (\mathbb{I}_i \xi_i)-u_i]$.  The dataset $\D_{i,n(t)}$ with $n: \R_{\geq 0}\to\N$ can change over time $t$, such that at time $t_1\in\R_{\geq 0}$ the dataset $\D_{i,n(t_1)}$ with $N(n(t_1))$ training points exists. This allows to accumulate training data over time, i.e., the number of training points $N(n)$ in the dataset $\D_{i,n}$ is monotonically increasing, but also ``forgetting" of training data to keep $N(n)$ constant. The time-dependent estimates of the GP are denoted by $\hat{\mathbf{f}}^{i,n}_{uk}(\mathbf{p}_i)$ to highlight the dependence on the corresponding dataset $\D_{i,n}$. Note that this construction also allows offline learning, i.e., the estimation depends only on the previously collected data, or any hybrid online/offline approach.


\textbf{Assumption 1:} The number of datasets $\D_{i,n}$ is finite and there are only finitely many switches of~$n(t)$ over time, such that there exists a time $T\in\R_{\geq 0}$ where $n(t)=n_{\text{end}},\forall t\geq T$.

Note that Assumption $1$ is little restrictive since the number of sets is often naturally bounded due to finite computational power or memory limitations and, since the unknown functions $\mathbf{f}_{uk}^{i}$ in \eqref{for:system1} are not explicitly time-dependent, long-life learning is typically not required. Therefore, there exists a constant dataset $\D_{i,n_{end}}$ for all $t>T_{end}$. Furthermore, Assumption $1$ ensures that the switching between the datasets is not infinitely fast which is natural in real-world applications.

Gaussian process models have been proven to be a very powerful oracle for nonlinear function regression. For the prediction, we concatenate the~$N(n)$ training points of $\D_{i,n}$ for each agent in an input matrix~$X_i=[\mathbf{p}_i^1,\mathbf{p}_i^2,\ldots,\mathbf{p}_i^{N(n)}]$ and a matrix of outputs~$Y_i^\top=[\mathbf{y}_i^1,\mathbf{y}_i^2,\ldots,\mathbf{y}_i^{N(n)}]$, where $\mathbf{y}$ might be corrupted by an additive Gaussian noise with $\mathcal{N}(0,\sigma I_6)$. A prediction for the output $\mathbf{y}^*_i\in\mathbb{R}^6$ at a new test point $\mathbf{p}_i^*\in SE(3)\times\mathfrak{se}(3)\simeq SE(3)\times\mathbb{R}^{6}$ is given by
\begin{align}
	\mean_i^{j}(\mathbf{y}_i^*\vert \mathbf{p}_i^*,\D_{i,n})&=m_i^{j}(\mathbf{p}_i^*)+\mathbf{k}_i(\mathbf{p}_i^*,X_i)^\top K_i^{-1}\label{for:gppred}\\
	&\phantom{=}\left(Y^{:,j}_i-[m_i^{j}(X^{:,1}_i),\ldots,m_i^{j}(X^{:,N}_i)]^\top\right)\notag\\
	\var_i^j(\mathbf{y}_i^*\vert \mathbf{p}_i^*,\D_{i,n})&=k_i(\mathbf{p}_i^*,\mathbf{p}_i^*)-\mathbf{k}_i(\mathbf{p}_i^*,X_i)^\top K^{-1}_i \mathbf{k}_i(\mathbf{p}_i^*,X_i)\notag
\end{align}
for all $j\in\{1,\ldots,6\}$, where $Y^{:,j}_i$ denotes the $j$-th column of the matrix of outputs~$Y_i$. The kernel $k_i\colon (SE(3)\times\mathbb{R}^{6})\times (SE(3)\times\mathbb{R}^{6})\to\R$ is a measure for the correlation of two states~$(\mathbf{p}_i,\mathbf{p}_i^\prime)$, whereas the mean function $m_i\colon SE(3)\times\mathbb{R}^{6}\to\R$ allows to include prior knowledge. The function~$K_i\colon (SE(3)\times\mathbb{R}^{6})^N\times  (SE(3)\times\mathbb{R}^{6})^N\to\R^{N\times N}$ is called the Gram matrix, whose elements are given by $K_i^{j',j}= k_i(X_i^{:, j'},X_i^{:, j})+\delta(j,j')\sigma^2$ for all $j',j\in\{1,\ldots,N\}$ with the delta function $\delta(j,j')=1$ for $j=j'$ and zero, otherwise. The vector-valued function~$\mathbf{k}_i\colon (SE(3)\times\mathbb{R}^{6})\times  (SE(3)\times\mathbb{R}^{6})^N\to\R^N$, with the elements~$k^j_i = k_i(\mathbf{p}_i^*,X_i^{:, j})$ for all $j\in\{1,\ldots,N\}$ and $i=1,\ldots,s$, expresses the covariance between $\mathbf{p}_i^*$ and the input training data $X_i$.

The selection of the kernel and the determination of the corresponding hyperparameters can be seen as degrees of freedom of the regression. The hyperparameters and the variance $\sigma$ of the Gaussian noise in the training data can be estimated by optimizing the marginal log-likelihood  (see~\cite{rasmussen2006gaussian} for instance). A powerful kernel for GP models of physical systems, and that we use in our simulations, is the squared exponential kernel, with an isotropic distance measure. An overview of the properties of different kernels can be found in~\cite{rasmussen2006gaussian}. In addition, the mean function can be achieved by common system identification techniques of the unknown dynamics~$\mathbf{f}_{uk}^{i}$ as described in~\cite{aastrom1971system}. However, without prior knowledge, the mean function is set to zero, i.e., $m_i(\mathbf{p}_i)=0$.

Based on \eqref{for:gppred}, the normal distributed components $(y_i^j)^*\vert \mathbf{p}_i^*,\D_{i,n}$ are combined into a multi-variable distribution for each agent 
$\mathbf{y}_i^*\vert(\mathbf{p}_i^*,\D_{i,n}) \sim \mathcal{N} (\Mean_i(\cdot),\Var_i(\cdot))$, where $\Mean_i(\mathbf{y}_i^*\vert \mathbf{p}_i^*,\D_{i,n})=[\mean_i^1(\cdot),\ldots,\mean_i^{6}(\cdot)]^\top$ and  $\Var_i(\mathbf{y}_i^*\vert \mathbf{p}_i^*,\D_{i,n})=\diag\left[\var_i^1(\cdot),\ldots,\var_i^{6}(\cdot)\right]$, with $i=1,\ldots,s$. 

\begin{remark}For simplicity, we consider identical kernels for each output dimension and each agent. However, the GP model can be easily adapted to different kernels for each output dimension and each agent. Moreover, as we use the GP in an online setting where new data is collected over time, the dataset $\D_n$ for the prediction \eqref{for:gppred} changes over time. The GP model allows to integrate new training data in a simple way by exploiting that every subset follows a multivariate Gaussian distribution.\end{remark}

For the later stability analysis of the closed-loop system, we introduce the following assumptions. In addition, we implicitly assume i.i.d data.

\textbf{Assumption 2:}	Consider a Gaussian process with the predictions $\hat{\mathbf{f}}_{uk}^{i,n}\in\mathcal{C}^0$ based on the dataset $\D_{i,n}$~\eqref{for:dataset}. Let $Q_\X\subset (\hbox{SE}(3)\times (\X\subset\R^6))$ be a compact set where $\hat{\mathbf{f}}_{uk}^{i,n}$ are bounded on $\X$. There exists a bounded function $\bar{\Delta}_{i,n}\colon Q_\X\to\R_{\geq 0}$ such that, the prediction error is bounded by
	\begin{align}
		\Prob\Bigg\lbrace\left\Vert\mathbf{f}_{uk}^{i}(\mathbf{p}_i)- \hat{\mathbf{f}}_{uk}^{i,n}(\mathbf{p}_i)\right\Vert\leq \mathbf{\bar{\rho}}_{i,n}(\mathbf{p}_i)\Bigg\rbrace\geq \delta_i
	\end{align}
	with probability $\delta_i\in (0,1]$ and $\mathbf{p}_i\in Q_\X$.

\begin{remark} Assumption $2$ ensures that on each dataset $\mathcal{D}_{i,n}$, there exists a probabilistic upper bound for the error between the prediction $\hat{\mathbf{f}}_{uk}^{i,n}(\mathbf{p}_i)$ and the actual $\mathbf{f}_{uk}^{i}(\mathbf{p}_i)$ on $Q_\X$.\end{remark}

To provide model error bounds, some assumptions on the unknown parts of the dynamics \eqref{for:system1} must be introduced \cite{wolpert1996lack}. 

\textbf{Assumption 3:}
	The kernel~$k_i$ is selected such that $\mathbf{f}_{uk}^{i}$ have a bounded reproducing kernel Hilbert space (RKHS) norm on $Q_\X$, i.e.,~$\Verts{\mathbf{f}_{uk}^{i,j}}_{k_i}<\infty$ for all~$j=1,\ldots 6$, and $i=1,\ldots,s$.

The norm of a function in a RKHS is a smoothness measure relative to a kernel~$k$ that is uniquely connected with this RKHS. In particular, it is a Lipschitz constant with respect to the metric of the used kernel. A more detailed discussion about RKHS norms is given in~\cite{wahba1990spline}.~Assumption $3$ requires that the kernel must be selected in such a way that the functions $\mathbf{f}_{uk}^{i}$ are elements of the associated RKHS. This sounds paradoxical since this function is unknown. However, there exist some kernels, namely universal kernels, which can approximate any continuous function arbitrarily precisely on a compact set~\cite[Lemma 4.55]{steinwart2008support}, such that the bounded RKHS norm is a mild assumption. Finally, with~Assumption $3$, the model error can be bounded as follows

\begin{lemma}[adapted from~\cite{srinivas2012information}]
	\label{lemma:boundederror}
Consider the unknown functions $\mathbf{f}_{uk}^{i}$ and a GP model satisfying~Assumption $3$. The model error for each agent is bounded  by
	\begin{align}
		\Prob\Bigg\lbrace\Bigg\Vert & \Mean_i\Big(\hat{\mathbf{f}}_{uk}^{i}(\mathbf{p}_i)\Big\vert \mathbf{p}_i,\D_{i,n}\Big)-\mathbf{f}_{uk}^{i}\Bigg\Vert\notag\\
		&\hspace{1.5cm}\leq \Bigg\Vert\bm{\beta}_{i,n}^\top \Var_i^{\frac{1}{2}}\Big(\hat{\mathbf{f}}_{uk}^{i}(\mathbf{p}_i)\Big\vert \mathbf{p}_i,\D_{i,n}\Big)\Bigg\Vert\Bigg\rbrace\geq \delta_i\notag
	\end{align}
	for each $i=1,\ldots,s$, and $\mathbf{p}_i\in Q_\mathcal{X},\delta_i\in(0,1)$ with $\bm{\beta}_{i,n}\in\mathbb{R}^6$, \begin{align}
    (\bm{\beta}_{i,n})_{j} =\sqrt{2\left\|\bar{\Delta}_{i,n}^{j}\right\|_{k}^{2}+300 \gamma_{i,j} \ln ^{3}\left(\frac{N(n)+1}{1-\delta^{1/6}}\right)}. 
\end{align}
The variable~$\gamma_{i,j} \in \R$ is the maximum information gain
\begin{align}
    \gamma_{i,j} &=\max _{\mathbf{p}_i^{\{1\}},\ldots, \mathbf{p}_i^{\{N(n)+1\}} \in Q_\X} \frac{1}{2} \log \left|I+\sigma_{i,j}^{-2} K_i\left(\mathbf{x}, \mathbf{x}^{\prime}\right)\right| \\
    \mathbf{x}, \mathbf{x}^{\prime} & \in\left\{\mathbf{p}_i^{\{1\}}, \ldots, \mathbf{p}_i^{\{N(n)+1\}}\right\}.
\end{align}
\end{lemma}

\begin{proof}
It is a direct implication of~\cite[Theorem 6]{srinivas2012information}.
\end{proof}

Note that the prediction error bound in~Assumption $2$ is given by 
$\bar{\Delta}_{i,n}(\mathbf{p}_i):= ||\bm{\beta}_{i,n}^\top\Var_i^{\frac{1}{2}}(\hat{\mathbf{f}}_{uk}^{i}(\mathbf{p}_i)\vert \mathbf{p}_i,\D_{i,n})||$ as shown by Lemma ~\ref{lemma:boundederror}.
\begin{remark}
An efficient algorithm can be used to find $\mathbf{\beta}_{i,n}$ based on the maximum information gain. Even though the values of the elements of~$\mathbf{\beta}_i$ are typically increasing with the number of training data, it is possible to learn the true functions $\mathbf{f}_{uk}^{i}$ arbitrarily exactly due to the shrinking variance $\Sigma_i$, 
see~\cite{Berkenkamp2016ROA}. In general, the prediction error bound $\bar{\Delta}_{i,n}(\mathbf{p}_i)$ is large if the uncertainty for the GP prediction is high and vice versa. Additionally, the bound is typically increasing if the set $Q_\X$ is expanded. The stochastic nature of the bound is because just a finite number of noisy training points are available.
\end{remark}

\subsection{Control design}
Consider the potential function $V$ as described in Section \ref{nominalsec}. In the absence of the unknown disturbances $V$ allows to write the closed loop system as \begin{align}\label{dynliegrpclosedloop}
\xi_i &= T_{g_i} L_{g^{-1}_i} \dot{g}_i,\\ \nonumber
\dot{\xi}_i&= - \mathbb{I}^{\sharp} \left(K_i T_{g_i} L_{g_i^{-1} } \left(\frac{\partial \psi_i}{\partial g_i}\right)\right) + \xi_i \theta_i\left(\xi_i, \frac{\partial \psi_i}{\partial t}\right) + F_d \xi_i.
\end{align} for $i=1,\ldots,s$. Next, we design a decentralized data-driven control law by using GPs to learn and mitigate the unknown disturbances of the system \eqref{for:system1}, and asymptotically stabilize the motion of the agents to a desired trajectory. 

\begin{theorem}\label{th2}
Consider the system \eqref{for:system1}, and a GP model trained  with the data set \eqref{for:dataset} satisfying Assumptions 1-3. Then the control law 
     \begin{align*} u_i &= - \mathbb{I}^{\sharp} \left(K_i T_{g_i} L_{g_i^{-1} } \left(\frac{\partial \psi_i}{\partial g_i}\right)\right) + \xi_i \theta_i\left(\xi_i, \frac{\partial \psi_i}{\partial t}\right) \\&+ F_d \xi_i - \mathbb{I}_i^\sharp ad^*_{\xi_i}( \mathbb{I}_i \xi_i)- \Mean_i\Big(\hat{\mathbf{f}}_{uk}^{i}(\mathbf{p}_i)\Big\vert \mathbf{p}_i,\D_{i,n}\Big)
  \end{align*}  where, $\xi_i = g^{-1}_i(t)g_i(t)$, $K_i \in \mathbb{R}^{+}$,  $F_d$  satisfies ${\llangle F_d \xi_i , \xi_i \rrangle}_{\mathbb{I}_i}\leq 0$, and $\theta_i(\xi_i, \frac{\partial \psi_i}{\partial t}) \in \mathbb{R}$ is given by 
  \[  \theta_i\left(\xi_i, \frac{\partial \psi_i}{\partial t}\right)=-\frac{c}{\tanh({\llangle \xi_i, \xi_i\rrangle}_{\mathbb{I}_i})} \Big| \frac{\partial \psi_i}{\partial t} \Big|,
  \]
    with $c > \hbox{max}_i\{K_i\}$, guarantees that the solution trajectories converge asymptotically to the desired equilibria set  $S$ and are ultimately uniformly bounded in probability on $Q_{\X}$ by \begin{align}
    \Prob\{||\gamma_{di}(t)||\leq\max_{\mathbf{p}_i\in Q_{\X}}\bar{\Delta}_{i,n(T)}(\mathbf{p}_i),\forall t\geq T_\epsilon\}\geq \epsilon,
\end{align} with $T_\epsilon\in\R_{\geq 0}$, where $\gamma_{d,i}$ is the tracking error, and $\bar{\Delta}_{i,n}(\mathbf{p}):Q_{\X}\to\R_{\geq 0}$ defines an upper bound of the model error for each $i=1,\ldots,s$.
\end{theorem}

\begin{proof}
    The proof of Theorem \ref{th2} follows similar computations than Theorem \ref{theorem}. By considering the Lyapunov function $\displaystyle{\sum_{i=1}^{s}K_i\psi_i+\frac{1}{2}\llangle\xi_i,\xi_i\rrangle_{\mathbb{I}_i}}$, the time derivative of $V$ along $u_i$ satisfies 
\begin{align*}\dot{V}&=\sum_{i=1}^{s}K_i\frac{\partial\psi_i}{\partial t}+\llangle\theta_i\xi_i,\xi_i\rrangle_{\mathbb{I}_i}+\llangle F_d\xi_i,\xi_i\rrangle_{\mathbb{I}_i}\\&+\left(\mathbf{f}_{uk}^{i}(\mathbf{p}_i)-\Mean_i\Big(\hat{\mathbf{f}}_{uk}^{i}(\mathbf{p}_i)\Big\vert \mathbf{p}_i,\D_{i,n}\Big)\right).\end{align*}

Since $\displaystyle{\sum_{i=1}^{s}K_i\frac{\partial\psi_i}{\partial t}+\llangle\theta_i\xi_i,\xi_i\rrangle_{\mathbb{I}_i}+\llangle F_d\xi_i,\xi_i\rrangle_{\mathbb{I}_i}\leq 0}$ by the proof of Theorem \ref{theorem}, and by employing Lemma
\ref{lemma:boundederror}, $$P\{\dot{V}\leq \bar{\Delta}_{i,n}(\mathbf{p}_i)\}\geq\epsilon,$$ where $\bar{\Delta}_{i,n}(t)(\mathbf{p}_i):Q_{\X}\to\R_{\geq 0}$ is a bounded function such that $\Vert\bm{\beta}_{i,n}^\top\Var_i^{\frac{1}{2}}(\hat{\mathbf{f}}_{uk}^{i}(\mathbf{p}_i)\vert \mathbf{p}_i,\D_{i,n}\Vert\leq\bar{\Delta}_{i,n}(\mathbf{p}_i)$, which exists because the kernel function is continuous and therefore, it is bounded on a compact set $Q_{\X}\subset(\hbox{SE}(3)\times (\X\subset\R^6))$, and then the variance $\Sigma(\hat{\mathbf{f}}_{uk}^{i}\mid \textbf{p}_i,\D_{i,n})$ is bounded (see \cite{beckers2016equilibrium}). Then, the value of $\dot{V}$ is negative with probability $\epsilon$ for all $\gamma_{d,i}$ with $\displaystyle{||\gamma_{d,i}||>\max_{\mathbf{p}_i\in Q_{\X}}\bar{\Delta}_{i,n}(\mathbf{p}_i)}$, where the maximum exists since $\bar{\Delta}_{i,n}(\mathbf{p}_i)$ is bounded in $Q_{\X}$. By using Assumption $1$, we define $T_{\epsilon}\in\R_{\geq 0}$ such that $\D_{n(T)}=\D_{n(t)}$ for all $t\geq T_{\epsilon}$. Then, $V$ is uniformly ultimately bounded in probability by $\displaystyle{\Prob\{||\gamma_{d,i}||\leq b,
\,\forall t\geq T_\epsilon\in\R_{\geq 0}\}\geq\epsilon}$, with probabilistic bound $\displaystyle{b=\max_{\mathbf{p}_i\in\Omega}\bar{\Delta}_{i,n(T)}(\mathbf{p}_i)}$.\end{proof}

\begin{remark}Note that the individual control law~$u_i(t)$ of each agent depends on the distance to its neighbors and the data set based on its own dynamics only. So, the proposed controller is decentralized.\end{remark}

\section{Simulation results}\label{sec5}

To validate the proposed control strategy, this section presents a simulation scenario involving seven UAVs flying in formation, each equipped with a camera for an aerial filming task. The control approach incorporates a potential function that serves two key purposes. First, it ensures collision avoidance by guiding the UAVs along safe trajectories. Additionally, the potential function includes a term that prevents the UAVs from adopting unfavorable orientations (see equation \eqref{potfun}). In this case, the parameter is set to $k = 1$, and the minimum inter-vehicle distance at which it is activated is 1.5 m. Specifically, each UAV adjusts its camera orientation to ensure that no other UAV obstructs its field of view. Such coordinated aerial filming tasks have been employed in multi-vehicle missions \cite{Moreno2021}.

The simulations were implemented in MATLAB Simulink. All UAVs have the same characteristics: a mass of approximately $1.3$ kg, a motor-to-motor distance (diagonal) of $450$ mm, and approximate moments of inertia $I_x \approx 0.02$ kg·m$^2$, $I_y \approx 0.02$ kg·m$^2$, and $I_z \approx 0.04$ kg·m$^2$. Additionally, sensor noise was introduced in the simulations, considering values consistent with those of a Pixhawk-based system. Specifically, Gaussian noise with a standard deviation of \(0.5^\circ\) was applied to the attitude measurements, while position estimates were affected by noise with a standard deviation of 1.0 meters. Furthermore, external disturbances were simulated, corresponding to wind gusts of up to 5 m/s.

\begin{figure}[h]
    \centering
    \includegraphics[width=1.0\textwidth]{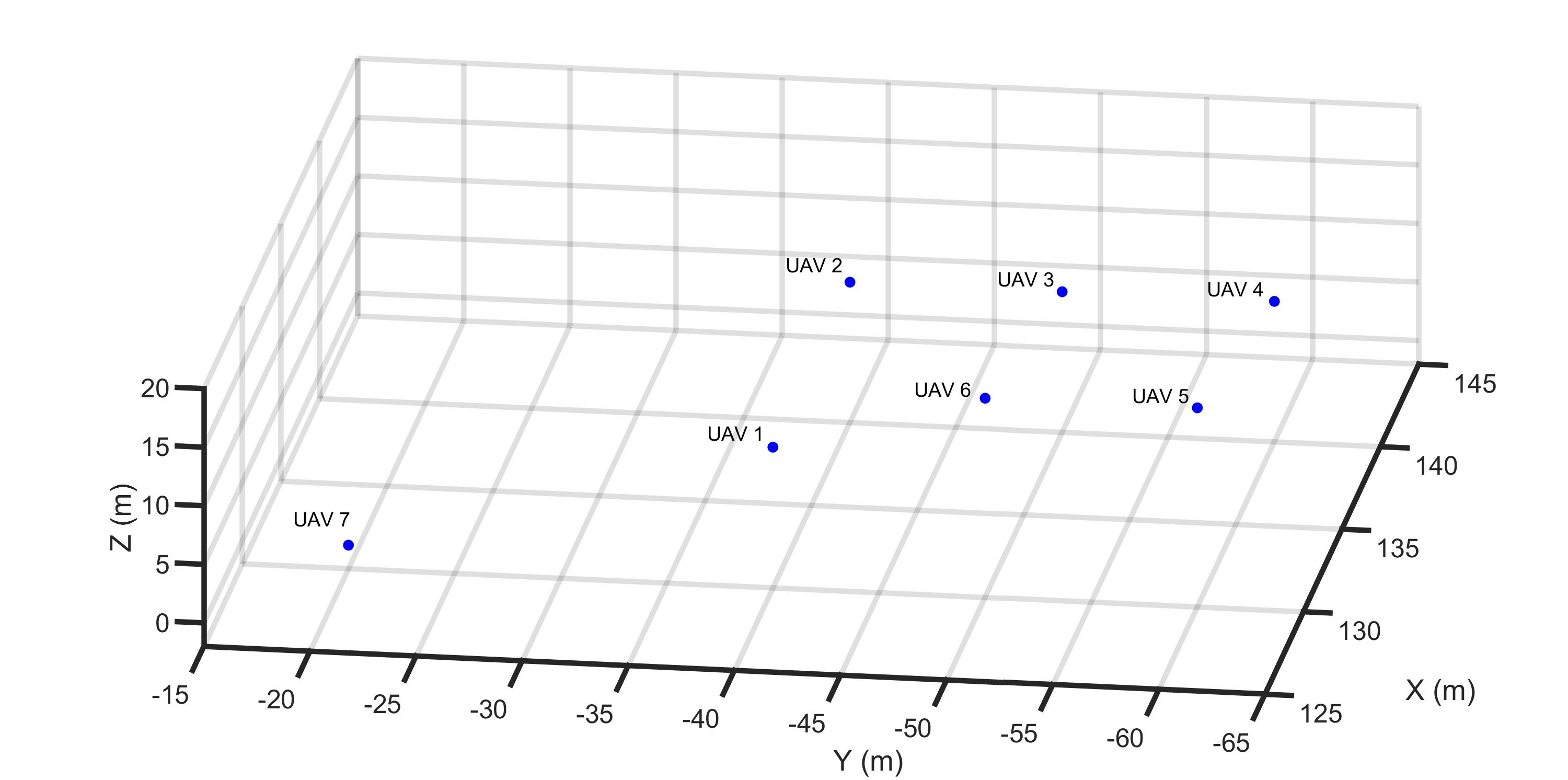}
    \caption{Spatial distribution of UAVs.}
    \label{fig:uavs}
\end{figure}

Figure \ref{fig:uavs} shows the spatial distribution of the UAVs. Initially, the vehicles are positioned in a square formation at different heights. The positions of the seven vehicles are: \( q_1=(130, -40, 10) \), \( q_2=(140, -40, 10) \), \( q_3=(140, -50, 10) \), \( q_4=(140, -60, 10) \), \( q_5=(130, -60, 15) \), \( q_6=(130, -50, 15) \), and the last vehicle is landed at \( q_7=(130, -20, 0) \).  

All seven vehicles have their cameras oriented along the x-axis. In this configuration, the UAV located at \( q_2 \) appears within the field of view (FOV) of the UAV at \( q_1 \). Consequently, the control algorithm adjusts the attitude of the UAV at \( q_1 \) to prevent filming the other UAV.  

Every eight seconds, the vehicles change their positions: each UAV at position \( q_i \) moves to position \( q_{i+1} \) for \( i = 1, \dots, 6 \), while the UAV at \( q_7 \) takes off and moves to \( q_1 \). Figure \ref{fig:uavs3D} shows the trajectory of the seven UAVs when moving from one point $q_i$ to the next one. 
During this process, the UAV at \( q_6 \) attempts to land at \( q_7 \), requiring the control algorithm to ensure collision avoidance.

\begin{figure}[h]
    \centering
    \includegraphics[width=1.0\textwidth]{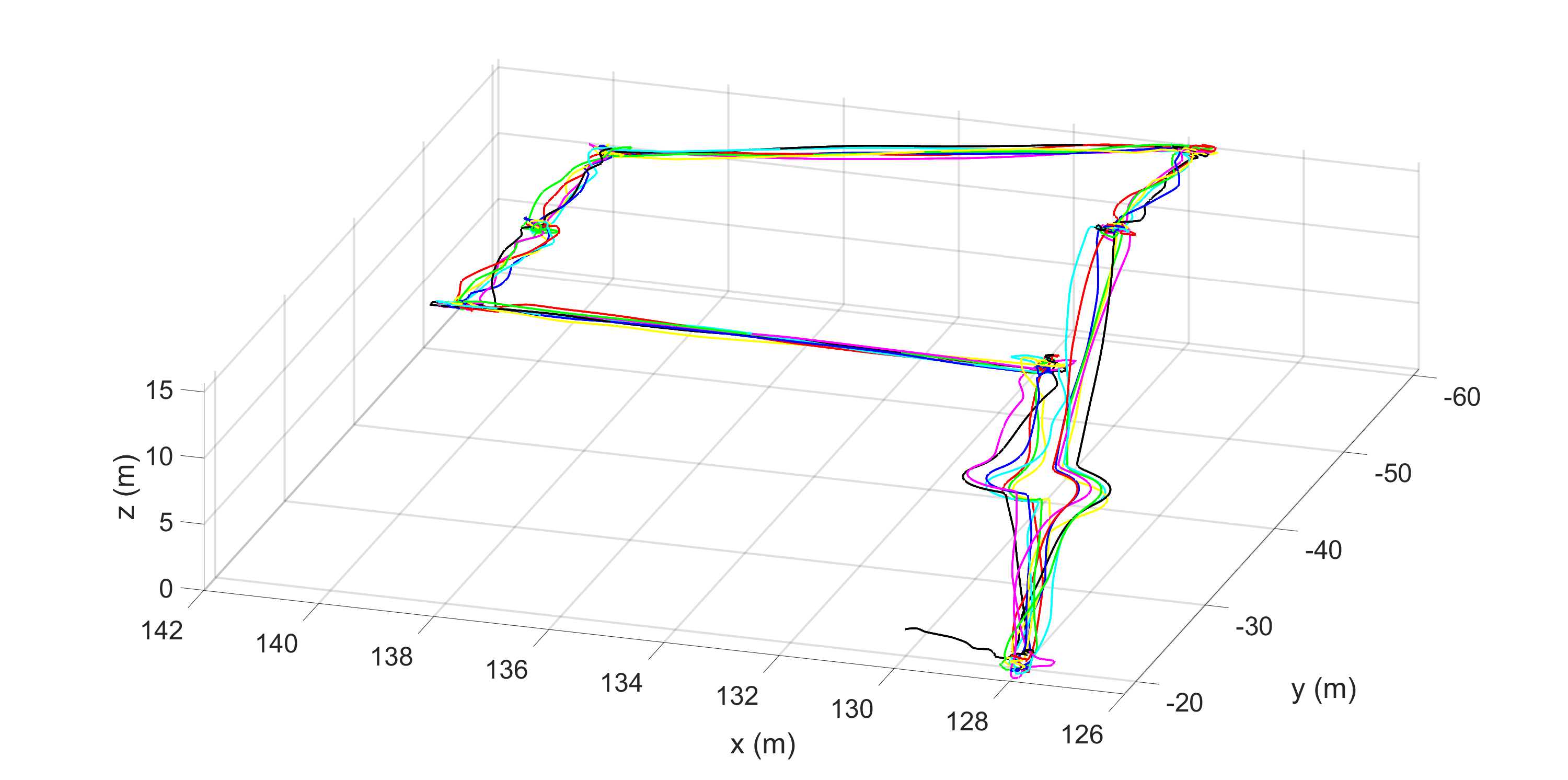}
    \caption{Trajectories of the seven UAVs in the formation. It can be observed how the vehicles evade each other during landing and takeoff.}
    \label{fig:uavs3D}
\end{figure}
 
\begin{figure}[h]
    \centering
    \includegraphics[width=1.0\textwidth]{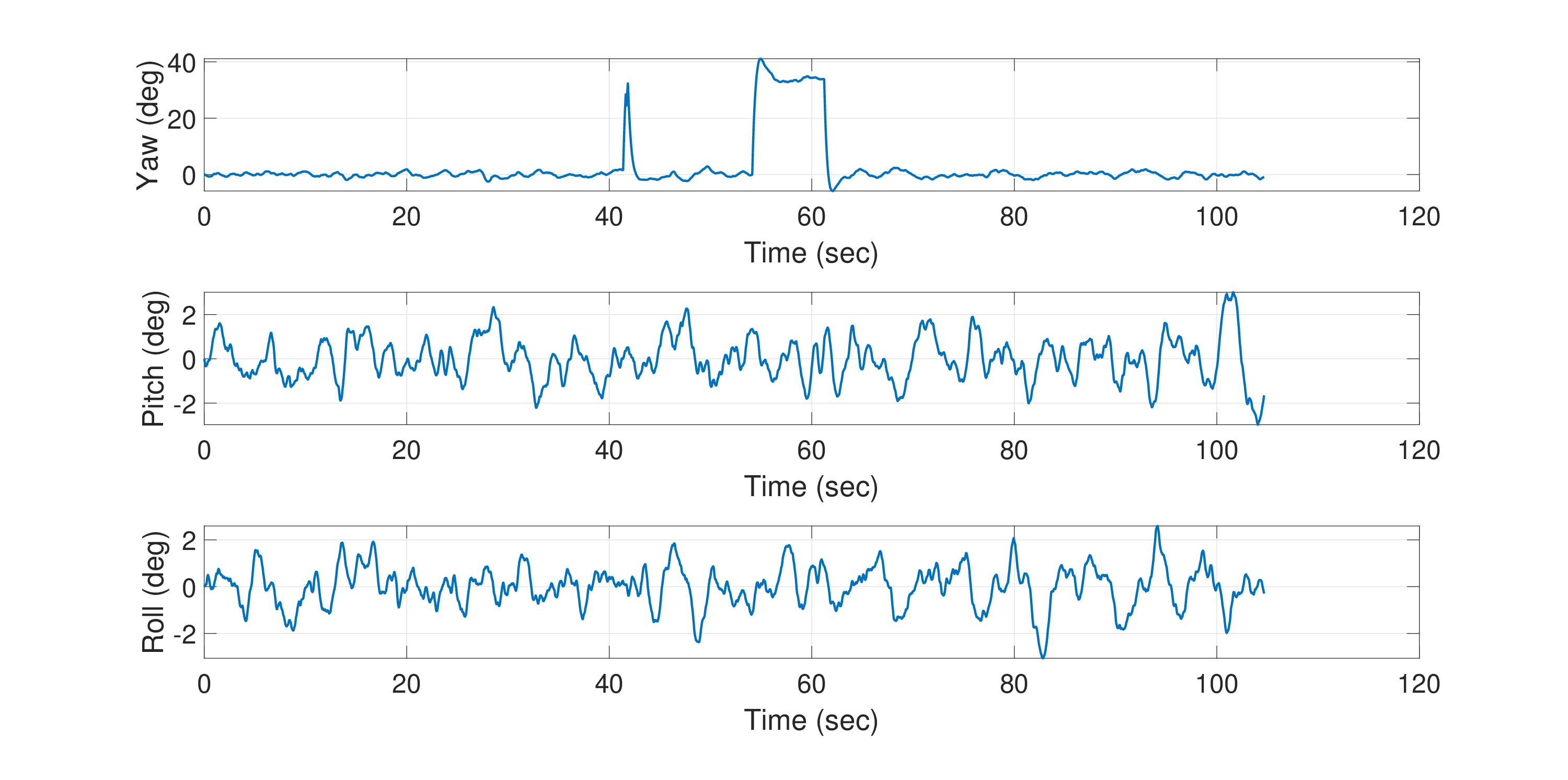}
    \caption{UAV attitude, avoiding filming other vehicles. The vehicle starts from position $q_2$ and goes through all points during the first $70$ seconds.}
    \label{fig:attcollision}
\end{figure}

\begin{figure}[h]
    \centering
    \includegraphics[width=1.0\textwidth]{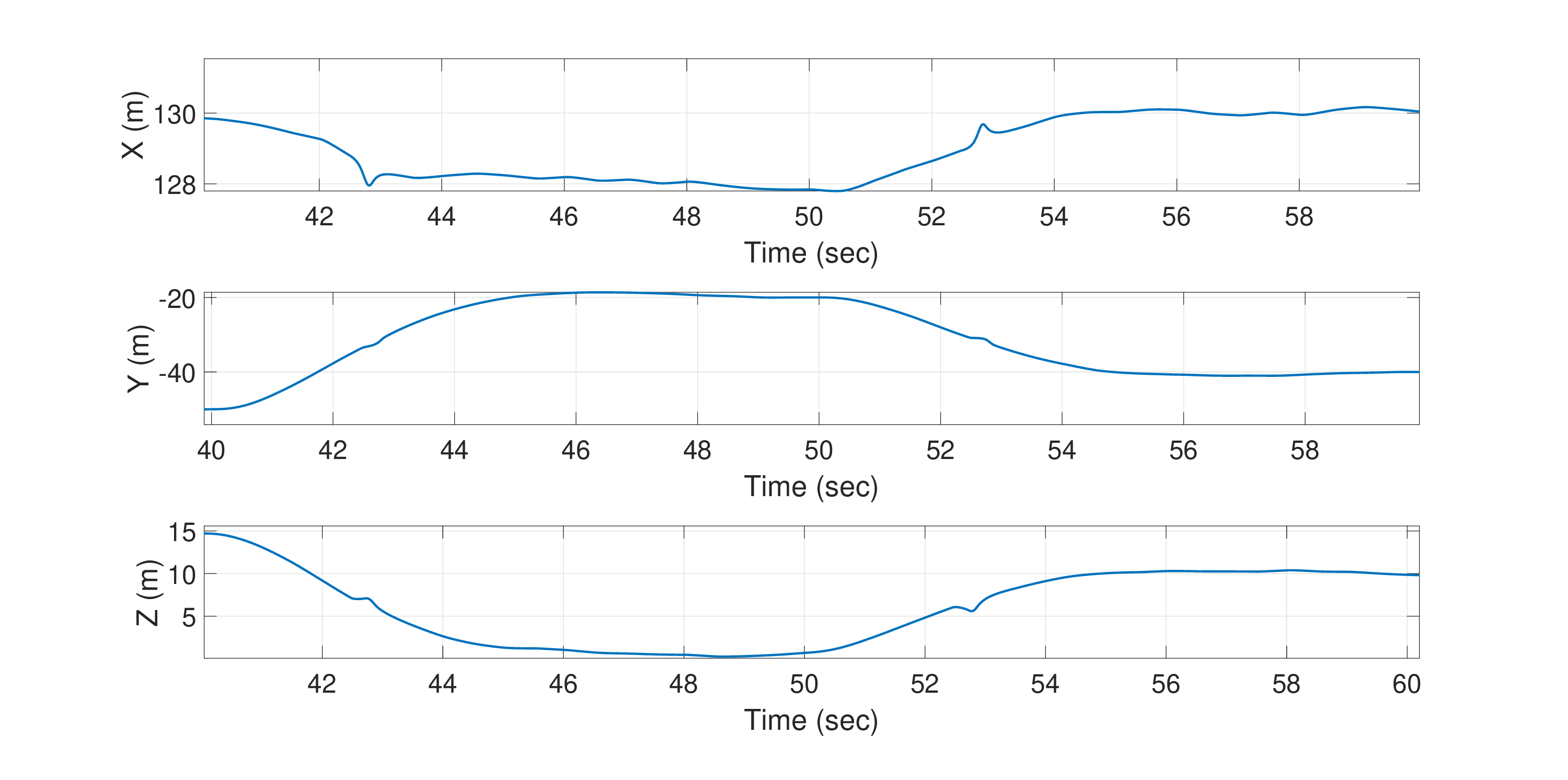}
    \caption{UAV position avoiding collision with other vehicles. The vehicle is at $q_{6}$ after $40$ seconds, lands at $q_{7}$ at $49$ seconds and passes by $q_{1}$ after $58$ seconds.}
    \label{fig:poscollision}
\end{figure}

In Figure \ref{fig:attcollision}, it can be observed that the attitude of the vehicle changes for approximately eight seconds when approaching position \( q_1 \) at 54 seconds, attempting to avoid filming the vehicle at position \( q_2 \). Additionally, at 41 seconds, the UAV briefly changes its orientation to avoid filming a passing vehicle.  

On the other hand, in Figure \ref{fig:poscollision}, at 42.5 seconds, the vehicle avoids colliding with another UAV when moving from position \( q_7 \) to \( q_1 \). Furthermore, at 52.5 seconds, when the UAV takes off from position \( q_7 \) to \( q_1 \), its trajectory intersects with that of another UAV, which is landing at position \( q_1 \). The control algorithm successfully prevents a collision. 

To further validate the effectiveness of incorporating Gaussian process (GP) estimation for disturbance compensation, the same mission is simulated under two conditions: one that includes GP-based estimation and one that does not. In these simulations, external disturbances are introduced to the UAVs, modeling the effects of factors such as wind gusts, ground effect variations, and sensor noise. By comparing the results, it is possible to evaluate how the GP-based approach improves trajectory tracking and orientation control, ensuring smoother and more stable flight behavior. The analysis highlights the advantages of using data-driven estimation to mitigate unmodeled disturbances, ultimately enhancing the robustness of the control strategy.

In Figures \ref{fig:attGP} and \ref{fig:posGP}, the response of a UAV in terms of orientation and position, respectively, can be observed. Around second 50, a torque and force disturbance is introduced into the vehicle, impacting tracking performance and even its ability to avoid collisions. This effect becomes evident when comparing the trajectory with that in Figure~\ref{fig:attcollision}. It can be seen that, without disturbances in the system, at second 54, the UAV adjusts its trajectory to avoid filming another vehicle, as previously mentioned. However, in Figure \ref{fig:attGP}, when the vehicle attempts the same maneuver, the disturbances prevent it from executing it correctly. Both in position and orientation, oscillations can be observed, affecting the UAV's navigation performance. These disturbances are estimated using a Gaussian Process (GP). Around second 80, compensation based on this estimation is introduced for both torque and force, and it can be observed that the oscillations disappear. In fact, at approximately second 125, the vehicle attempts the same maneuver it tried at second 54, but this time it successfully executes it, thanks to the GP-based compensation.

As explained previously, the GP model is used to compensate for the unknown dynamics. This GP model works with a dataset of the form \eqref{for:dataset} constructed using data related to the states $\bm{p}_{i}$, consisting of $n=250$ points, and also requires a kernel. We have decided to use the well-known Squared Exponential kernel because it provides good performance. In Figures \ref{fig:attGPest} and \ref{fig:posGPes}, the torque and force distributions can be observed. The figures show the perturbation in a red line, the estimation in a dashed blue line, and the confidence interval in gray.


\begin{figure}[h]
    \centering
    \includegraphics[width=1.0\textwidth]{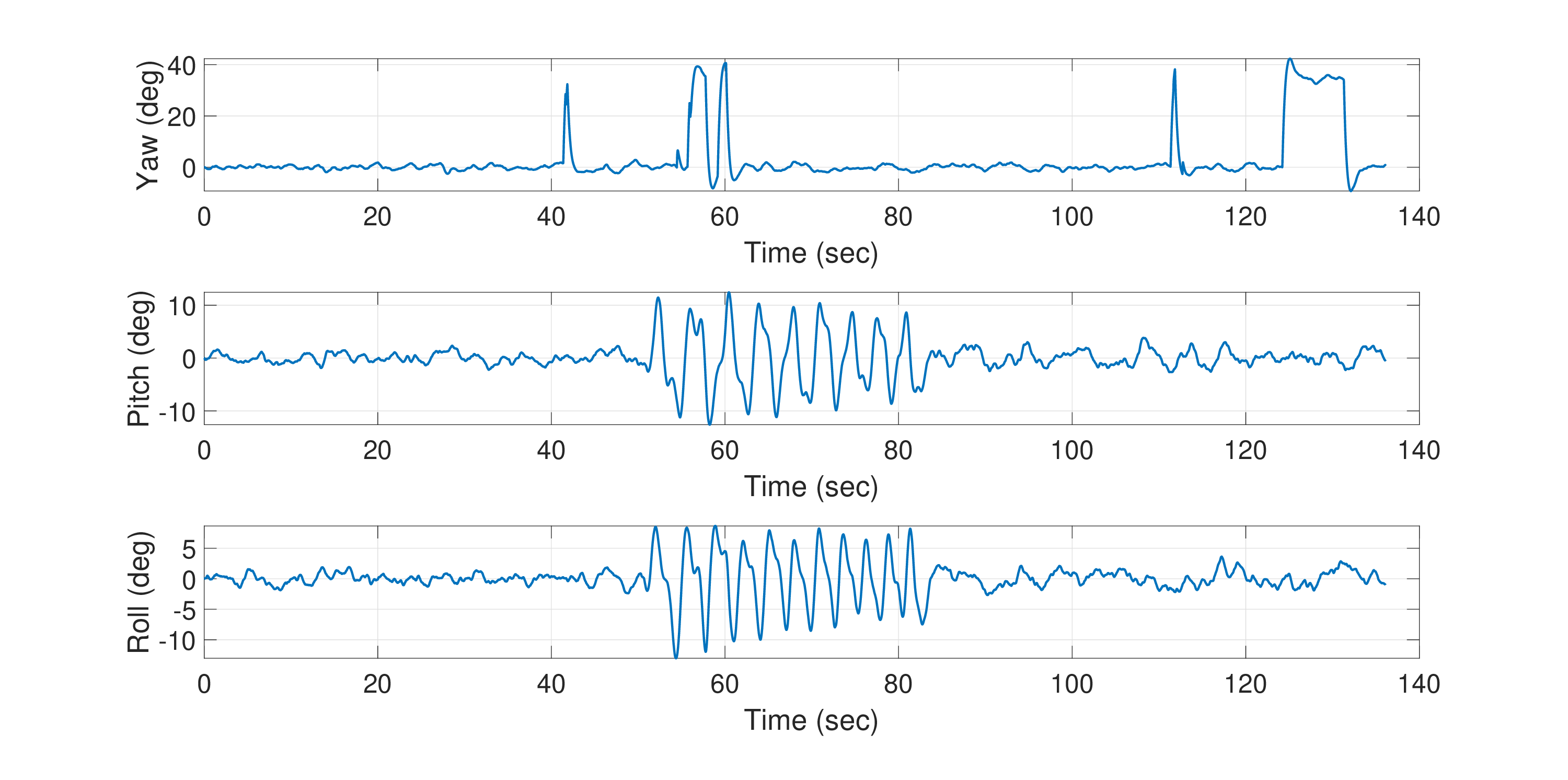}
    \caption{UAV attitude response with and without GP compensation. The vehicle starts again from position $q_2$ and goes through all points during the first $70$ seconds. Around second $80$, the GP estimation is used.}
    \label{fig:attGP}
\end{figure}

\begin{figure}[h]
    \centering
    \includegraphics[width=1.0\textwidth]{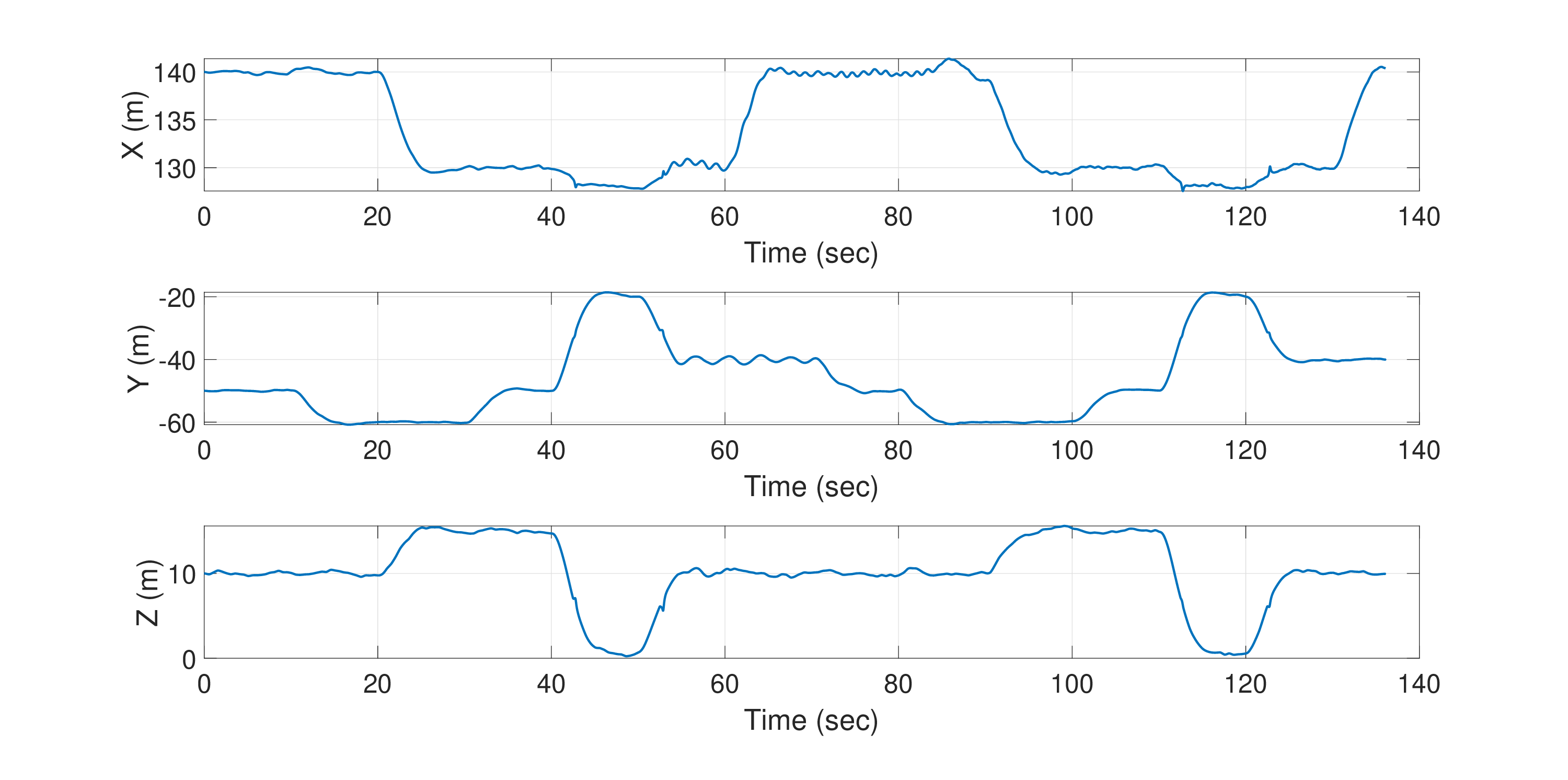}
    \caption{UAV position response with and without GP compensation. We are observing position during the same trajectory as in Figure \ref{fig:attGP}.}
    \label{fig:posGP}
\end{figure}

\begin{figure}[h]
    \centering
    \includegraphics[width=1.0\textwidth]{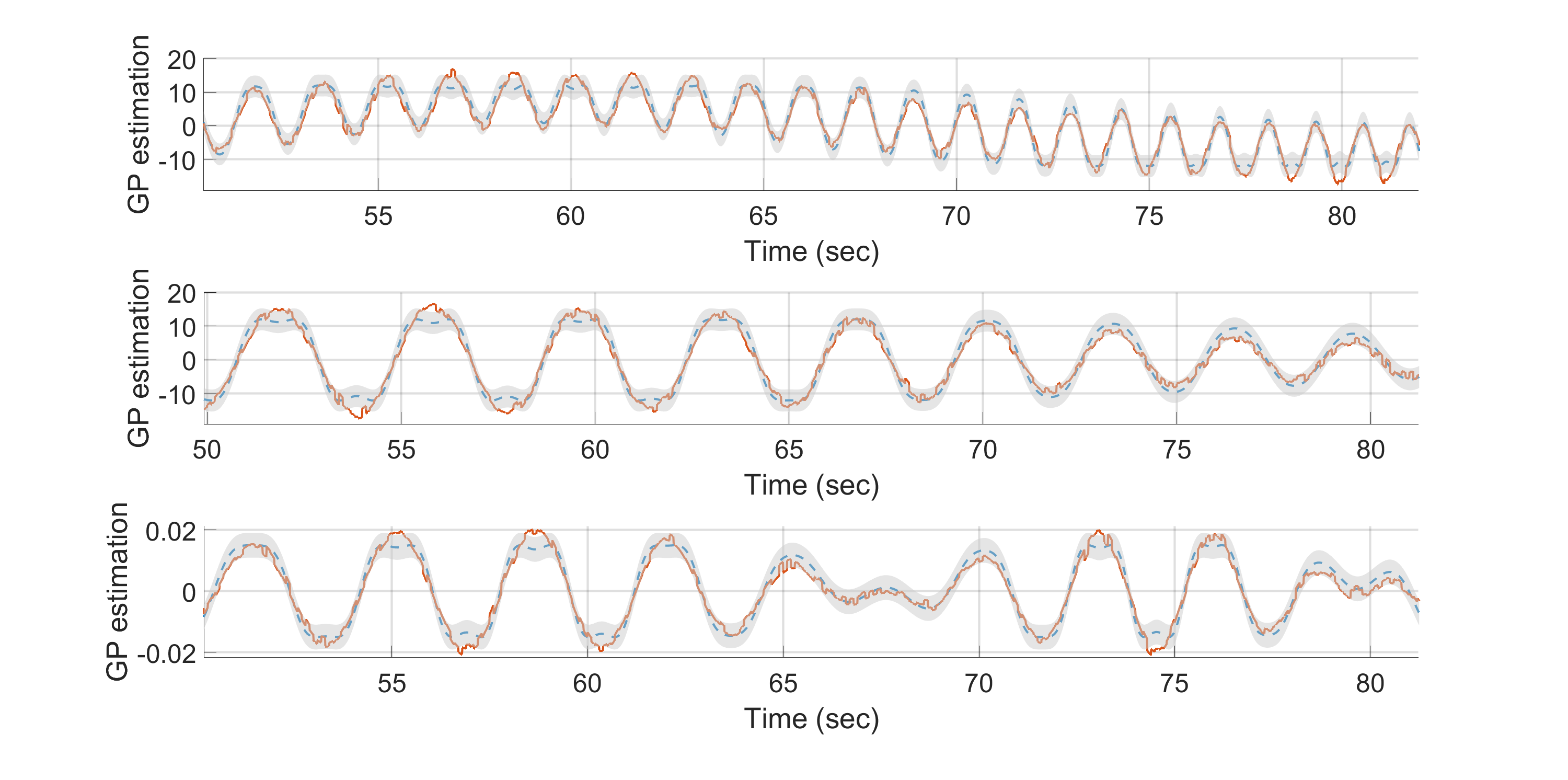}
    \caption{GP-based estimation of torque disturbances.}
    \label{fig:attGPest}
\end{figure}

\begin{figure}[h]
    \centering
    \includegraphics[width=1.0\textwidth]{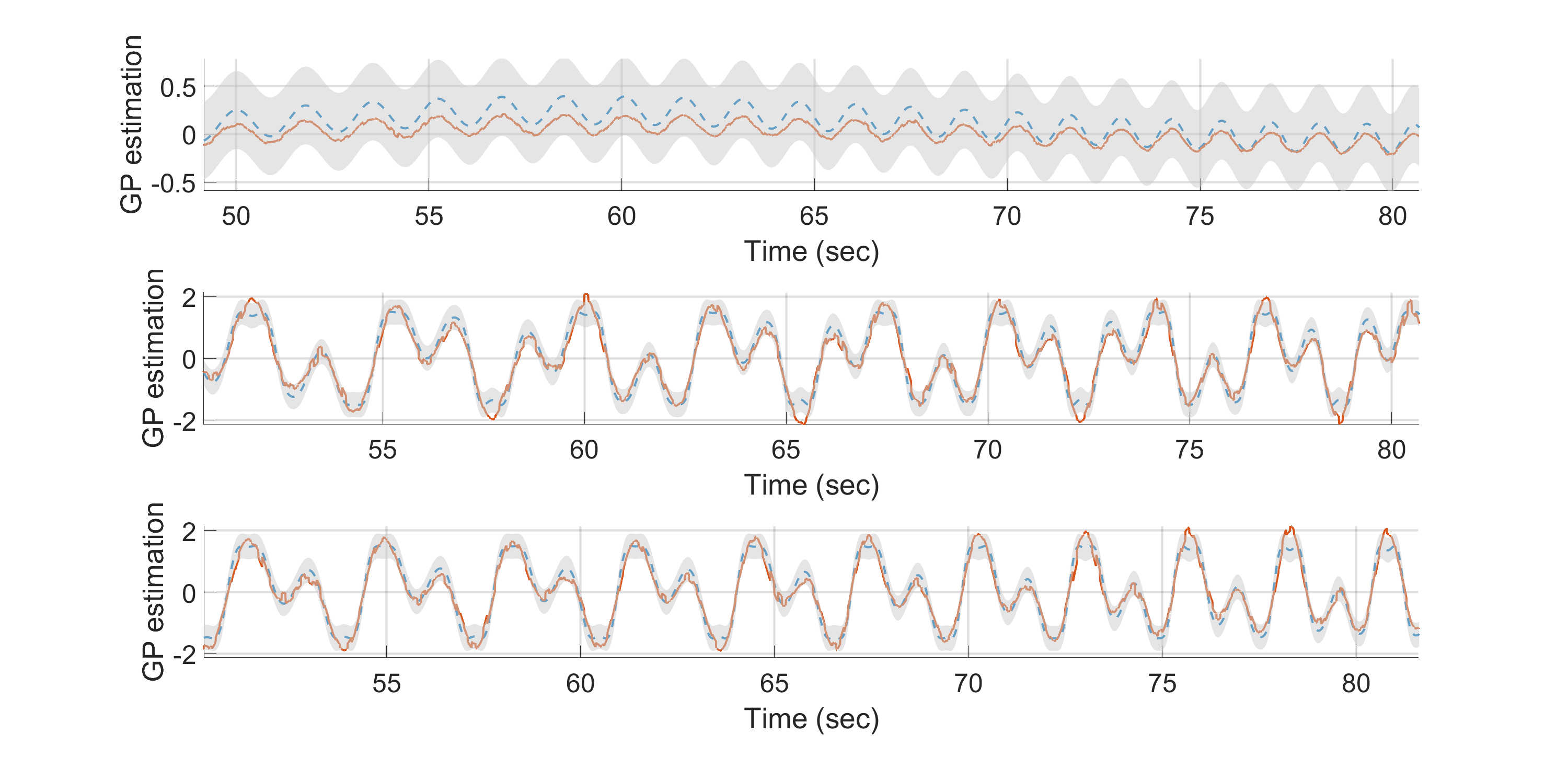}
    \caption{GP-based estimation of force disturbances.}
    \label{fig:posGPes}
\end{figure}

\section{Conclusions and Future Work}

This paper introduced a learning-based decentralized control strategy for multi-agent systems evolving on the Lie group $SE(3)$, incorporating Gaussian processes (GPs) to handle unknown dynamics while ensuring collision avoidance and trajectory tracking. By leveraging GPs, the proposed control approach provides probabilistic guarantees on tracking performance, improving the robustness and adaptability of multi-agent systems in uncertain environments.

The key contributions of this work include: (i) The integration of learning-based control with decentralized navigation functions to enhance adaptability in multi-agent coordination.
(ii) A mathematically rigorous framework that ensures probabilistic bounded tracking errors under partially unknown dynamics.
(iii) Simulation results demonstrating the effectiveness of the proposed controller in achieving precise trajectory tracking while avoiding collisions. The results confirm that the proposed approach outperforms traditional decentralized control methods by dynamically adapting to unknown disturbances, making it highly relevant for applications in robotic swarms and collaborative robotic tasks.

Several directions remain open for future research, including (i) extending the model to scenarios where agents are underactuated and operate in time-varying, uncertain environments, including obstacles and dynamically changing conditions, as well as implementing real-world experiments to validate the method beyond simulation. (ii)  Examining the impact of adversarial noise and sensor failures on learning-based control performance, and
developing control laws with stability guarantees under more general uncertainty models.

\end{document}